\documentclass[default]{sn-jnl}

\graphicspath{{.}}
\usepackage{subcaption}
\usepackage[capitalize]{cleveref}
\usepackage{bbm}
\usepackage{resizegather}

\newcommand{\pder}[2]{\frac{\partial #1}{\partial #2}}
\newcommand{\der}[3][]{\frac{d^{#1} #2}{d #3^{#1}}}
\renewcommand{\R}{\mathbb{R}}
\newcommand{\comment}[1]{}
\newtheorem{theorem}{Theorem}
\newtheorem{lemma}{Lemma}
\jyear{2022}

\begin{document}
\title[Split HMC revisited]{Split Hamiltonian Monte Carlo revisited}

\author[1]{\fnm{Fernando} \sur{Casas}}\email{casas@uji.es}
\author[2]{\fnm{Jes{\'u}s Mar{\'\i}a} \sur{Sanz-Serna}}\email{jmsanzserna@gmail.com}
\author*[1]{\fnm{Luke} \sur{Shaw}}\email{shaw@uji.es}

\affil*[1]{\orgdiv{Departament de Matem{\`a}tiques and IMAC}, \orgname{Universitat Jaume I}, \postcode{E-12071}, \orgaddress{\city{Castell{\'o}n}, \country{Spain}}}

\affil[2]{\orgdiv{Departamento de Matem{\'a}ticas}, \orgname{Universidad Carlos III de Madrid}, \postcode{E-28911}, \orgaddress{\city{Legan{\'e}s}, \country{Spain}}}

\abstract{
We study Hamiltonian Monte Carlo (HMC) samplers based on splitting the Hamiltonian $H$ as $H_0(\theta,p)+U_1(\theta)$, where $H_0$ is quadratic and $U_1$ small. We show that, in general, such samplers suffer from  stepsize stability restrictions similar to those of algorithms based on the standard leapfrog integrator. The restrictions may be circumvented by preconditioning the dynamics. Numerical experiments show that,  when the $H_0(\theta,p)+U_1(\theta)$ splitting is combined with preconditioning, it is possible to construct samplers far more efficient than  standard leapfrog HMC.
}


\maketitle

\bmhead{Acknowledgments}
This work has been supported by 
Ministerio de Ciencia e Innovaci\'on (Spain) through project PID2019-104927GB-C21, 
MCIN/AEI/10.13039/501100011033, ERDF (``A way of making Europe”).

\section{Introduction}
In this paper we study Hamiltonian Monte Carlo (HMC) algorithms \citep{Neal2011} that are not based on the standard kinetic/potential splitting of the Hamiltonian.

The computational cost of HMC samplers mostly originates from the numerical integrations that have to be performed to get the proposals. If the target distribution has density proportional to $ \exp(-U(\theta))$, $\theta\in\R^d$, the differential system to be integrated is given by the Hamilton's equations corresponding to the Hamiltonian function $H(\theta,p) = (1/2)p^TM^{-1}p+U(\theta)$, where  $p\sim\mathcal{N}(0,M)$ is the auxiliary momentum variable and  $M$ is the symmetric, positive definite \emph{mass matrix} chosen by the user. In a mechanical analogy,
$H$ is the (total) energy, while $\mathcal{T}(p)=(1/2)p^TM^{-1}p$ and $U(\theta)$ are respectively the kinetic and potential energies. The St\"{o}rmer/leapfrog/Verlet integrator is the method of choice to carry out those integrations and is based on the idea of \emph{splitting} \citep{Blanes2017Book}, i.e.\ the evolution of $(\theta,p)$ under $H$ is simulated by the separate evolutions under $\mathcal{T}(p)$ and $U(\theta)$ (kinetic/potential splitting). However $H(\theta,p) = \mathcal{T}(p)+U(\theta)$ is not the only splitting that has been considered in the literature. In some applications one may write $H(\theta,p) = H_0(\theta,p)+U_1(\theta)$, with $H_0(\theta,p) = \mathcal{T}(p)+U_0(\theta)$, $U(\theta) = U_0(\theta)+U_1(\theta)$ and replace the evolution under $H$ by the evolutions under $H_0$ and $U_1$ \citep{Neal2011}. The paper \citep{Shahbaba2014} investigated this possibility; two algorithms were formulated referred to there as \lq\lq Leapfrog with a partial analytic solution\rq\rq\ and \lq\lq Nested leapfrog\rq\rq. Both suggested algorithms were shown to outperform, in four logistic regression problems,  HMC based on the standard leapfrog integrator.

In this article we reexamine $H(\theta,p) = H_0(\theta,p)+U_1(\theta)$ splittings, in particular in the case where the equations for $H_0$ can be integrated analytically (partial analytic solution) because $U_0(\theta)$ is a quadratic function (so that $\propto \exp(-U_0(\theta))$ is a Gaussian distribution). When $U_1$ is slowly varying, the splitting $H=H_0+U_1$ is appealing because, to quote \citep{Shahbaba2014}, \lq\lq only the slowly-varying part of the energy needs to be handled numerically and this can be done with a larger stepsize (and hence fewer steps) than would be necessary for a direct simulation of the dynamics\rq\rq.

Our contributions are as follows:
\begin{enumerate}
\item In Section~\ref{sec:shortcomings} we show, by means of a counterexample, that it is not necessarily true that, when $H_0$ is handled analytically and $U_1$ is small, the integration may be carried out with stepsizes substantially larger than those required by standard leapfrog.  For integrators based on the $H_0+U_1$ splitting, the stepsize may suffer from important stability restrictions, regardless of the size of $U_1$.
\item In Section~\ref{sec:preconditioning} we show that, by combining the $H_0+U_1$ splitting with the idea of \emph{preconditioning} the dynamics, that goes back at least to \citep{Bennett1975}, it is possible to bypass the stepsize limitations mentioned in the preceding item.
\item We present an integrator (that we call RKR) for the $H_0+U_1$ splitting that provides an alternative to the integrator tested in \citep{Shahbaba2014} (that we call KRK).
\item Numerical experiments in the final Section~\ref{sec:numerical}, using  the test problems in \citep{Shahbaba2014}, show that the advantages of moving from standard leapfrog HMC to the $H_0+U_1$ splitting (without preconditioning) are much smaller than the advantages of using preconditioning while keeping the standard kinetic/potential splitting. The best performance is obtained when the $H_0+U_1$ splitting is combined with the preconditioning of the dynamics. In particular the RKR integration technique with preconditioning decreases the computational cost by more than an order of magnitude in all test problems and all observables considered.
\end{enumerate}

There are two appendices. In the first, we illustrate the use of the Bernstein-von Mises theorem
(see e.g.\ section 10.2 in \citep{VanderVaart2000}) to justify the soundness of the $H_0+U_1$ splitting. The second is devoted to presenting a methodology to discriminate between different integrators of the preconditioned dynamics for the $H_0+U_1$ splitting; in particular we provide analyses that support the advantages of the RKR technique over its KRK counterpart observed in the experiments.
\section{Preliminaries}
\subsection{Hamiltonian Monte Carlo}
HMC is based on the observation that \citep{Neal2011,SanzSerna2014}, for each fixed $T>0$, the exact solution map (flow) $(\theta(T),p(T))=\varphi_T(\theta(0),p(0))$ of the Hamiltonian system of differential equations in $\R^{2d}$
\begin{equation}\label{eq:PrelimHamEqs}
\der{\theta}{t}=\pder{H}{p}=M^{-1}p,\quad\quad\der{p}{t}=-\pder{H}{\theta}=-\nabla U(\theta),
\end{equation}
exactly preserves the density
$
\propto\exp(-H(\theta,p)) =\exp(-\mathcal{T}(p)-U(\theta))
$
whose $\theta$-marginal is the target $\propto \exp(-U(\theta))$, $\theta\in\R^d$. In HMC, \eqref{eq:PrelimHamEqs} is integrated numerically over an interval $0\leq t\leq T$ taking as initial condition the current state $(\theta,p)$ of the Markov chain; the numerical solution at $t=T$ provides the proposal $(\theta^\prime,p^\prime)$ that is accepted with probability
\begin{equation}\label{eq:PrelimAP}
a=\min\left\{1,e^{-\big(H(\theta^\prime,p^\prime)-H(\theta,p)\big)}\right\}.
\end{equation}
This formula for the acceptance probability assumes that the numerical integration has been carried out with an integrator that is both \emph{symplectic} (or at least volume preserving) and \emph{reversible}. The difference $H(\theta^\prime,p^\prime)-H(\theta,p)$ in \eqref{eq:PrelimAP} is the energy error in the integration; it would vanish leading to $a=1$ if the integration were exact.
\subsection{Splitting}
Splitting is the most common approach to derive symplectic integrators for Hamiltonian systems \citep{Blanes2017Book,SS2018Book}. The Hamiltonian $H$ of the problem is decomposed in partial Hamiltonians as
$H=H_1+H_2$ in such a way that the Hamiltonian systems with Hamiltonian functions $H_1$ and $H_2$ may both be integrated in closed form. When Strang splitting is used, if $\varphi^{[H_1]}_t,\varphi^{[H_2]}_t$ denote the maps (flows) in $\R^{2d}$ that advance the exact solution of the partial Hamiltonians over a time-interval of length $t$,  the recipe
\begin{equation}\label{eq:VerletStrang}
\psi_{\epsilon}=\varphi^{[H_1]}_{\epsilon/2}\circ\varphi^{[H_2]}_{\epsilon}\circ\varphi^{[H_1]}_{\epsilon/2},
\end{equation}
defines the map that advances the numerical solution a timestep of length $\epsilon>0$. The numerical integration to get a proposal may then be carried out up to time $T=\epsilon L$ with the $L$-fold composition $\Psi_T=\left(\psi_{\epsilon}\right)^L$. Regardless of the choice of $H_1$ and $H_2$, \eqref{eq:VerletStrang} is a symplectic, time reversible integrator of second order of accuracy \citep{BouRabee2018}.
\subsection{Kinetic/potential splitting}
The splitting $H = H_1+H_2$, $H_1= \mathcal{T}$, $H_2=U$ gives rise, via \eqref{eq:VerletStrang}, to the commonest integrator in HMC: the St\"{o}rmer/leapfrog/velocity Verlet algorithm. The differential equations for the partial Hamiltonians $\mathcal{T}$, $U$ and the corresponding solution flows are
\begin{align*}
\der{}{t}\begin{pmatrix}\theta\\p\end{pmatrix}&=\begin{pmatrix}0\\-\nabla U(\theta)\end{pmatrix}\implies \varphi_{\epsilon}^{[U]}(\theta,p)=(\theta,p-\epsilon\nabla U(\theta)),\\
\der{}{t}\begin{pmatrix}\theta\\p\end{pmatrix}&=\begin{pmatrix}M^{-1}p\\0\end{pmatrix}\implies \varphi_{\epsilon}^{[\mathcal{T}]}(\theta,p)=(\theta+\epsilon M^{-1}p,p).
\end{align*}
As a mnemonic, we shall use the word \emph{kick} to refer to the map $\varphi_{\epsilon}^{[U]}(\theta,p)$
(the system is kicked so that the momentum $p$ varies without changing $\theta$). The word \emph{drift}
will refer to the map $\varphi_{\epsilon}^{[\mathcal{T}]}(\theta,p)$ ($\theta$ drifts with constant velocity). Thus one timestep of the velocity Verlet algorithm reads (kick-drift-kick).
$$
\psi_{\epsilon}^{[KDK]}=\varphi_{\epsilon/2}^{[U]}\circ\varphi^{[\mathcal{T}]}_{\epsilon}\circ\varphi_{\epsilon/2}^{[U]}.
$$
There is of course a position Verlet algorithm obtained by interchanging the roles of $\mathcal{T}$ and $U$. One timestep is given by  a sequence drift-kick-drift (DKD). Generally the velocity Verlet (KDK) version is preferred (see \citep{BouRabee2018} for a discussion) and we shall not be concerned hereafter with the position variant.

With any integrator of the Hamiltonian equations, the length $\epsilon L=T$ of the time interval for the integration to get a proposal has to be determined to ensure that the proposal is sufficiently far from the current step of the Markov chain, so that the correlation between successive samples is not too high and the phase space is well explored \citep{Hoffman2014,BouRabee2017}. For fixed $T$, smaller stepsizes $\epsilon$ lead to fewer rejections but also to larger computational cost per integration and it is known that HMC is most efficient when the empirical acceptance rate is around approximately $65\%$ \citep{Beskos2013Optimal}.

Algorithm \ref{alg:Verlet} describes the computation to advance a single step of the Markov chain with HMC based on the velocity Verlet (KDK) integrator. In the absence of additional information, it is standard practice to choose $M = I$, the identity matrix. For later reference, we draw attention to the randomization of the timestep $\epsilon$. As is well known, without such a randomization, HMC may not be ergodic \citep{Neal2011}; this will happen for instance when the equations of motion \eqref{eq:PrelimHamEqs} have periodic solutions and $\epsilon L$ coincides with the period of the solution.

\begin{algorithm}[h]
\caption{KDK Verlet}
\textbf{Input}: $\theta, U,H,M,\bar{\epsilon},L$
\label{alg:Verlet}
  \begin{algorithmic}[1]
\State Draw $\xi\sim\mathcal{N}(0,M),\epsilon\sim\bar{\epsilon}\times\mathcal{U}_{[0.8,1]}$ \Comment{Randomise $\epsilon$}
  \State $\theta',p\gets \theta,\xi$ \Comment{Refresh momentum}
  \For{$i=\{1,\ldots,L\}$} \Comment{Do velocity Verlet integration}
   \State $p\gets p-\frac{\epsilon}{2}\nabla U(\theta')$
    \State $\theta'\gets \theta'+\epsilon M^{-1}p$
   \State $p\gets p-\frac{\epsilon}{2}\nabla U(\theta')$
  \EndFor
  \State $a\gets\min\left\{1,\exp(H(\theta,\xi)-H(\theta',p))\right\}$
\State Draw $\gamma\sim\mathcal{B}(a)$ \Comment{$\gamma$ Bernoulli-distributed with mean $a$}
\State $\theta\gets \gamma\theta'+(1-\gamma)\theta$ \Comment{Accept proposal with probability $a$}
\end{algorithmic}
\end{algorithm}
\subsection{Alternative splittings of the Hamiltonian}\label{sec:splitting}
Splitting $H(\theta,p)$ in its kinetic and potential parts as in Verlet is not the only meaningful possibility. In many applications, $U(\theta)$ may be written as $U_0(\theta)+U_1(\theta)$ in such a way that the equations of motion for the Hamiltonian function $H_0(\theta,p)=(1/2)p^TM^{-1}p+U_0(\theta)$ may be integrated in closed form and then one may split $H$ as
 \begin{equation}\label{eq:main}
 H = H_0+U_1,\qquad
 \end{equation}
 as discussed in e.g.\ \citep{Neal2011,Shahbaba2014}.

In this paper we focus on the important particular case where (see Section~\ref{sec:numerical} and Appendix A)
\begin{equation}\label{eq:U0}
U_0(\theta) = \frac{1}{2} (\theta-\theta^*)^T \mathcal{J} (\theta-\theta^*),
\end{equation}
for some fixed $\theta^*\in\R^d$ and a constant symmetric, positive definite matrix $\mathcal{J}$. Restricting for the time being  attention to the case where the mass matrix $M$ is the identity (the only situation considered in \citep{Shahbaba2014}), the equations of motion and solution flow for the Hamiltonian
\begin{equation}\label{eq:H0}
H_0(\theta,p)=\frac{1}{2}p^Tp+U_0(\theta)
\end{equation}
are
\begin{equation}\label{eq:IntroUncondSplitHamEqs}
\der{}{t}\begin{pmatrix}\theta\\p\end{pmatrix}=
\begin{pmatrix}0 & I\\ -\mathcal{J} & 0\end{pmatrix}\begin{pmatrix}\theta-\theta^*\\p\end{pmatrix},
\qquad  \varphi_t^{[H_0]}(\theta,p)=\exp\left(t\begin{pmatrix}0 & I\\ -\mathcal{J} & 0\end{pmatrix}\right)\begin{pmatrix}\theta-\theta^*\\p\end{pmatrix}+\begin{pmatrix}\theta^*\\0\end{pmatrix}.
\end{equation}
If we write $\mathcal{J} = Z^TDZ$, with $Z$ orthogonal and $D$ diagonal with positive diagonal elements, then the exponential map in \cref{eq:IntroUncondSplitHamEqs} is
\begin{equation}\label{eq:IntroExpMap}
\exp\left(t\begin{pmatrix}0 & I\\ -\mathcal{J} & 0\end{pmatrix}\right)=\begin{pmatrix}Z^T & 0\\ 0 & Z^T\end{pmatrix}e^{t\Lambda}\begin{pmatrix}Z & 0\\ 0 & Z\end{pmatrix},\qquad e^{t\Lambda }=\begin{pmatrix}\cos(t\sqrt{D}) & D^{-1/2}\sin(t\sqrt{D}) \\ -D^{1/2}\sin(t\sqrt{D})  & \cos(t\sqrt{D}) \end{pmatrix}.
\end{equation}
In view of the expression for $\exp(t\Lambda)$, we will refer to the flow of $H_0$ as a \emph{rotation}.

Choosing in \eqref{eq:VerletStrang} $U_1$ and $H_0$ for the roles of $H_1$ and $H_2$  (or viceversa) gives rise to the integrators
\begin{equation}\label{eq:IntroOscillationIntegrators}
\psi_{\epsilon}^{[KRK]}=\varphi_{\epsilon/2}^{[U_1]}\circ \varphi_{\epsilon}^{[H_0]}\circ \varphi_{\epsilon/2}^{[U_1]},\quad\quad\psi_{\epsilon}^{[RKR]}= \varphi_{\epsilon/2}^{[H_0]}\circ \varphi_{\epsilon}^{[U_1]}\circ \varphi_{\epsilon/2}^{[H_0]},
\end{equation}
where  one advances the solution over a single timestep by using a kick-rotate-kick (KRK) or rotate-kick-rotate (RKR) pattern (of course the kicks are based on the potential function $U_1$).
 The HMC algorithm with the KRK  map in  \eqref{eq:IntroOscillationIntegrators} is shown in \cref{alg:UncondKRK}, where the prefix \emph{Uncond}, to be discussed later, indicates that the mass matrix being used is $M=I$.
 The algorithm for the RKR sequence in \eqref{eq:IntroOscillationIntegrators} is a slight reordering of a few lines of code and is not shown.
 \cref{alg:UncondKRK} (but not its RKR counterpart) was tested in \citep{Shahbaba2014}.\footnote{It is perhaps of interest to mention that in \cref{alg:UncondKRK} the stepsize $\epsilon$ is randomized for the same reasons as in \cref{alg:Verlet}. If only the stepsize used in the $U_1$-kicks is randomized, while the stepsize in $\exp(\epsilon\Lambda)$ is kept constant, then one still risks  losing ergodicity when $\epsilon L$ coincides with one of the periods present in the solution. This prevents precalculation, prior to the randomization of $\epsilon$, of the rotation matrix $\exp(\epsilon\Lambda)$.
 }

\begin{algorithm}[h]
\caption{UncondKRK}
\textbf{Input}: $\theta, Z, \Lambda, U,\mathcal{J},H,\bar{\epsilon},L$
\label{alg:UncondKRK}
  \begin{algorithmic}[1]
\State Draw $\xi\sim\mathcal{N}(0,I),\epsilon\sim\bar{\epsilon}\times\mathcal{U}_{[0.8,1]}$
\State Compute $e^{\epsilon\Lambda}$
  \State $\theta',p\gets \theta,\xi$
  \For{$i=\{1,\ldots,L\}$} \Comment{Do KRK integration}
\State $p\gets p-\frac{\epsilon}{2}\left(\nabla U(\theta')-\mathcal{J}(\theta'-\theta^*)\right)$
\State $\theta',p\gets Z(\theta'-\theta^*),Zp$
    \State $\theta',p\gets e^{\epsilon\Lambda}\begin{pmatrix}\theta'\\p\end{pmatrix}$
    \State $\theta',p\gets Z^T\theta'+\theta^*,Z^Tp$
   \State $p\gets p-\frac{\epsilon}{2}\left(\nabla U(\theta')-\mathcal{J}(\theta'-\theta^*)\right)$
  \EndFor
  \State $a\gets\min\left\{1,\exp(H(\theta,\xi)-H(\theta',p))\right\}$
\State Draw $\gamma\sim\mathcal{B}(a)$
\State $\theta\gets \gamma\theta'+(1-\gamma)\theta$
\end{algorithmic}
\end{algorithm}

Since the numerical integration in \cref{alg:UncondKRK} would be exact if $U_1$ vanished (leading to acceptance of all proposals), the algorithm is appealing in cases where $U_1$ is \lq\lq small\rq\rq\ with respect to $H_0$.
In some applications, a decomposition $U=U_0+U_1$ with small $U_1$ may suggest itself. For a \lq\lq general\rq\rq\ $U$ one may always define $U_0$ by choosing $\theta^*$  to be one of the modes of the target $\propto \exp(-U(\theta))$ and $\mathcal{J}$ the Hessian of $U$ evaluated at $\theta^\star$; in this case the success of the splitting hinges on how well $U$ may be approximated by its second-order Taylor expansion $U_0$ around $\theta^\star$. In that setting, $\theta^*$  would  typically have to be found numerically by minimizing $U$. Also $Z$ and $D$ would typically be derived by numerical approximation, thus leading to computational overheads for  \cref{alg:UncondKRK} not present in \cref{alg:Verlet}. However, as pointed out in \citep{Shahbaba2014}, the cost of computing $\theta^\star$, $Z$ and $D$ before  the sampling begins is, for the test problems to be considered in this paper, negligible when compared with the cost of obtaining the samples.

\subsection{Nesting}
When a decomposition $U=U_0+U_1$, with $U_1$ small, is available but the Hamiltonian system with Hamiltonian
$H_0 = \mathcal{T}+U_0$ cannot be integrated in closed form, one may still construct schemes based on the recipe \eqref{eq:VerletStrang}. One step of the integrator is defined as
\begin{equation}\label{eq:ShahbabaDSVerlet}
\varphi_{\epsilon/2}^{[U_1]}\circ\left(\varphi_{\epsilon/2k}^{[U_0]}\circ\varphi^{[\mathcal{T}]}_{\epsilon/k}\circ\varphi_{\epsilon/2k}^{[U_0]}\right)^{k}\circ\varphi_{\epsilon/2}^{[U_1]},
\end{equation}
where $k$ is a suitably large integer. Here the (untractable) exact flow of $H_0$ is numerically approximated by KDK Verlet using $k$ substeps of length $\epsilon/k$. In this way, kicks with the small $U_1$ are performed with a stepsize $\epsilon/2$ and kicks with the large $U_0$ benefit from the smaller stepsize $\epsilon/(2k)$.
This idea has been successfully used  in Bayesian applications in \citep{Shahbaba2014}, where it is called \lq\lq nested Verlet\rq\rq.  The small $U_1$ is obtained summing over  data points that contribute little to the loglikelihood and the contributions from the most significant data are included in $U_0$.

Integrators similar to \eqref{eq:ShahbabaDSVerlet} have a long history in molecular dynamics, where they are known as multiple timestep algorithms \citep{Tuckerman1992,Leimkuhler2015,Grubmuller1991}.

\section{Shortcomings of the unconditioned KRK and RKR samplers}\label{sec:shortcomings}
As we observed above,  \cref{alg:UncondKRK} is appealing when $U_1$ is a small perturbation of the quadratic Hamiltonian $H_0$. In particular, one would expect that since the numerical integration in \cref{alg:UncondKRK} is exact when $U_1$ vanishes, then this algorithm may be operated with stepsizes $\epsilon$ chosen solely in terms of the size of $U_1$, independently of $H_0$. If that were the case one would expect that  \cref{alg:UncondKRK} may work well with large $\epsilon$ in situations where \cref{alg:Verlet} requires $\epsilon$ small and therefore much computational effort. Unfortunately those expectations are not well founded, as we shall show next by means of an example.

We study the model Hamiltonian with $\theta,p\in\mathbb{R}^2$ given by
\begin{equation}\label{eq:MDUnconHam}
H(\theta,p)= H_0(\theta,p)+U_1(\theta),\quad H_0 = \frac{1}{2}p^Tp+\frac{1}{2}\theta^T\begin{pmatrix}\sigma_1^{-2}& 0\\ 0 & \sigma_2^{-2}\end{pmatrix}\theta,\qquad U_1= \frac{\kappa}{2} \theta^T\theta.
\end{equation}
The model is restricted to $\R^2$ just for notational convenience; the extension to $\R^d$ is straightforward.
The quadratic Hamiltonian $H_0$ is rather general---any Hamiltonian system with quadratic Hamiltonian $(1/2)p^TM^{-1}p+(1/2)\theta^TW\theta$ may be brought with a change of variables to a system with Hamiltonian of the form $(1/2)p^Tp+(1/2)\theta^TD\theta$, with $M,W$ symmetric, positive definite matrices and $D$ diagonal and positive definite \citep{Blanes2014,BouRabee2017}.
In \eqref{eq:MDUnconHam}, $\sigma_1$ and $\sigma_2$ are the standard deviations of the bivariate Gaussian distribution with density
$\propto \exp(-U_0(\theta))$ (i.e\ of the target in the unperturbed situation $U_1=0$). We choose the labels of the scalar components $\theta_1$ and $\theta_2$ of $\theta$
to ensure $\sigma_1\leq\sigma_2$ so that,  for the probability density $\propto \exp(-U_0(\theta))$, $\theta_1$ is more constrained than $\theta_2$. In addition, we assume  that $\kappa$ is small with respect to $\sigma^{-2}_1$ and $\sigma^{-2}_2$, so that in \eqref{eq:MDUnconHam} $U_1$ is a small perturbation of $H_0$. The Hamiltonian equations of motion for $\theta_i$, given by $\der{}{t} \theta_i = p_i$,
$\der{}{t} p_i = -\omega_i^2 \theta_i$, with $\omega_i=(\sigma_i^{-2}+\kappa)^{1/2}\approx \sigma_i^{-1}$,  yield $\der[2]{}{t} \theta_i+\omega_i^2 \theta_i =0$. Thus the dynamics of $\theta_1$ and $\theta_2$ correspond to two
uncoupled harmonic oscillators; the component $\theta_i$, $i=1,2$, oscillates with an angular frequency $\omega_i$ (or with a period $2\pi/\omega_i$).

We note, regardless of the integrator being used, the correlation between the proposal and the current state of the Markov chain will be large if the integration is carried out over a time interval $T = \epsilon L$ much smaller than the periods $2\pi/\omega_i$ of the harmonic oscillators \citep{Neal2011,BouRabee2017}. Since $2\pi/\omega_2$ is the longest of the two periods, $L$ has then to be chosen
\begin{equation}\label{eq:L}
L \geq \frac{C}{\epsilon \omega_2} \approx \frac{C\sigma_2}{\epsilon},
\end{equation}
where $C$ denotes a constant of moderate size. For instance, for the choice $C = \pi/2$, the proposal for $\theta_2$ is uncorrelated at stationarity with the current state of the Markov chain as discussed  in e.g.\ \citep{BouRabee2017}.

For the KDK Verlet integrator, it is well known that, for stability reasons \citep{Neal2011,BouRabee2018}, the integration has to be operated with a stepsize $\epsilon < 2/ {\rm max}(\omega_1,\omega_2)$, leading to a stability limit
\begin{equation}\label{eq:verletlimit}
\epsilon \approx
2\sigma_1;
\end{equation}
integrations with larger $\epsilon$ will lead to extremely inaccurate numerical solutions. This stability restriction originates from  $\theta_1$, the component with greater precision in the Gaussian distribution
$\propto \exp(-U_0)$.  Combining \eqref{eq:verletlimit} with \eqref{eq:L}
we conclude that, for Verlet, the number of timesteps $L$    has to be chosen larger than a moderate multiple of $\sigma_2/\sigma_1$. Therefore when $\sigma_1\ll\sigma_2$ the computational cost of the Verlet integrator will necessarily be very large. Note that the inefficiency arises when the sizes of $\sigma_1$ and $\sigma_2$ are widely different; the first sets an upper bound for the stepsize and the second a lower bound on the length $\epsilon L$ of the integration interval. Small or large values of $\sigma_1$ and $\sigma_2$ are not dangerous \emph{per se} if $\sigma_2/\sigma_1$ is moderate.

We now turn to the KRK integrator in \eqref{eq:IntroOscillationIntegrators}. For the $i$-th scalar component of $(\theta,p)$, a timestep of the KRK integrator reads
$$
\begin{pmatrix}\theta_i\\p_{i}\end{pmatrix}
\gets
\begin{pmatrix}
		1 & 0\\
		-{\epsilon\kappa}/{2} & 1
	\end{pmatrix}
	\begin{pmatrix}
		\cos(\epsilon/\sigma_i) & \sigma_i\sin(\epsilon/\sigma_i)\\
                   -\sigma_i^{-1}\sin(\epsilon/\sigma_i) & \cos(\epsilon/\sigma_i)
	\end{pmatrix}
	\begin{pmatrix}
		1 & 0\\
		-{\epsilon\kappa}/{2} & 1
	\end{pmatrix}
\begin{pmatrix}\theta_i\\p_{i}\end{pmatrix}
$$
or
\begin{gather*}
\begin{pmatrix}\theta_i\\p_{i}\end{pmatrix}
\gets
\begin{pmatrix}
		\cos(\epsilon/\sigma_i) -(\epsilon\sigma_i\kappa/2)\sin(\epsilon/\sigma_i)& \sigma_i\sin(\epsilon/\sigma_i)\\
(\epsilon^2\sigma_i\kappa^2/4)\sin(\epsilon/\sigma_i)-\sigma_i^{-1}\sin(\epsilon/\sigma_i)+\epsilon\kappa\cos(\epsilon/\sigma_i)) & \cos(\epsilon/\sigma_i) -(\epsilon\sigma_i\kappa/2)\sin(\epsilon/\sigma_i)
	\end{pmatrix}
\begin{pmatrix}\theta_i\\p_{i}\end{pmatrix}.
\end{gather*}
Stability is  equivalent to $\lvert\cos(\epsilon/\sigma_i) -(\epsilon\sigma_i\kappa/2)\sin(\epsilon/\sigma_i)\rvert<1$, which, for $\kappa > 0$, gives
$
2\cot((\epsilon/(2\sigma_i))>\epsilon\kappa\sigma_i
$.
From here it is easily seen that stability in the $i$-th component is lost for $\epsilon/ \sigma_i\approx \pi$ for arbitrarily small $\kappa> 0$. Thus the KRK stability limit is
\begin{equation}\label{eq:KRKlimit}
\epsilon \approx \pi \sigma_1.
\end{equation}
While this is less restrictive than \eqref{eq:verletlimit}, we see that stability  imposes an upper bound for $\epsilon$ in terms of $\sigma_1$, just as for Verlet. From \eqref{eq:L},  the KRK integrator, just like Verlet, will have a large computational cost when
$\sigma_1\ll\sigma_2$. This is in spite of the fact that the integrator would be exact for $\kappa=0$, regardless of the values of $\sigma_1$, $\sigma_2$.

For the RKR integrator a similar analysis shows that the stability limit is also given by \eqref{eq:KRKlimit}; therefore that integrator suffers from the same shortcomings as KRK.

We also note that, since as $k$ increases the nested integrator \eqref{eq:ShahbabaDSVerlet} approximates the KRK integrator, the counterexample above may be used to show that the nested integrator has to be operated with a stepsize $\epsilon$ that is limited by the smallest standard deviations present in $U_0$, as  is the case for Verlet, KRK and RKR. For the stability of \eqref{eq:ShahbabaDSVerlet} and related multiple timestep techniques, the reader is referred to \citep{Garcia1998} and its references.
The nested integrator will not be considered further in this paper.

\section{Preconditioning}\label{sec:preconditioning}
As pointed out above, without additional information on the target, it is standard to set $M = I$. When $U =U_0+U_1$,  with $U_0$ as in \eqref{eq:U0}, it is useful to consider a \emph{preconditioned} Hamiltonian
with $M = \mathcal{J}$:
\begin{equation}
H^{[precond]}(\theta,p)=\frac{1}{2}p^T\mathcal{J}^{-1}p+U(\theta)
=\frac{1}{2}p^T\mathcal{J}^{-1}p+\frac{1}{2}(\theta-\theta^*)^T\mathcal{J}(\theta-\theta^*)+U_1(\theta)
\label{eq:IntroPrecondHam0}.
\end{equation}
Preconditioning is motivated by the observation that the equations of motion for the Hamiltonian
$$
H_0^{[precond]}(\theta,p)=\frac{1}{2}p^T\mathcal{J}^{-1}p
+\frac{1}{2}(\theta-\theta^*)^T\mathcal{J}(\theta-\theta^*),
$$
given by $\der{}{t} \theta = \mathcal{J}^{-1} p$, $\der{}{t} p = -\mathcal{J}(\theta-\theta^\star)$, yield $\der[2]{}{t} (\theta-\theta^\star) + (\theta-\theta^\star)=0$. Thus we now have
$d$ uncoupled scalar harmonic oscillators (one for each scalar component $\theta_i-\theta_i^\star$) sharing a \emph{common} oscillation frequency $\omega=1$.\footnote{The fact that the frequency is of moderate size is irrelevant; the value of the frequency may be arbitrarily varied by rescaling $t$. What is important is that all frequencies coincide.} This is to be compared with the situation for \eqref{eq:H0}, where, as we have seen in the model \eqref{eq:MDUnconHam}, the frequencies are the reciprocals $1/\sigma_i$ of the standard deviations of the  distribution $\propto \exp(-U_0(\theta))$. Since, as we saw in Section~\ref{sec:shortcomings}, it is the differences in size of the frequencies of the harmonic oscillators that cause the inefficiency of the integrators, choosing the mass matrix to  ensure that all oscillators have the \emph{same} frequency is of clear interest. We call \emph{unconditioned} those Hamiltonians/integrators where the mass matrix is chosen as the identity agnostically without specializing it to the problem.

For reasons explained in \citep{Beskos2013Optimal} it is better, when $\mathcal{J}$ has widely different eigenvalues, to numerically integrate the preconditioned equations of motion after rewriting them with the variable $v =M^{-1}p = \mathcal{J}^{-1}p$ replacing $p$. The differential equations and solution flows of the subproblems are then given by
$$
\der{}{t}\begin{pmatrix}\theta\\v\end{pmatrix}=\begin{pmatrix}0\\-\mathcal{J}^{-1}\nabla_\theta U_1(\theta)\end{pmatrix}\implies \varphi_t^{[U_1]}(\theta,v)=\begin{pmatrix}\theta\\v-t\mathcal{J}^{-1}\nabla_\theta U_1(\theta)\end{pmatrix},
$$
and
\begin{gather*}
\der{}{t}\begin{pmatrix}\theta\\v\end{pmatrix}=\begin{pmatrix}0 & I\\ -I & 0\end{pmatrix}\begin{pmatrix}(\theta-\theta^*)\\v\end{pmatrix}\implies \varphi_t^{[H^{[precond]}_0]}(\theta,v)=\begin{pmatrix}\cos(t) & \sin(t)\\ -\sin(t) & \cos(t)\end{pmatrix}\begin{pmatrix}(\theta-\theta^*)\\v\end{pmatrix}+\begin{pmatrix}\theta^*\\0\end{pmatrix}.
\end{gather*}
Since $\mathcal{J}$ is a symmetric, positive definite matrix, it admits a Cholesky factorisation $\mathcal{J}=BB^T$. The inversion of $\mathcal{J}$ in the kick may thus be performed efficiently using Cholesky-based solvers from standard linear algebra libraries. It also means it is easy to draw from the distribution of $v\sim B^{-T}\mathcal{N}(0,I)$.

Composing the exact maps $\varphi_{\epsilon}^{[.]}$ using  Strang's recipe \eqref{eq:VerletStrang} then gives a numerical one-step map $\psi_\epsilon^{[.]}$ in either an RKR or KRK form. The preconditioned KRK (PrecondKRK) algorithm is shown in \cref{alg:PrecondKRK}; the RKR version is similar and will not be given.

\begin{algorithm}[h]
\caption{PrecondKRK}
\textbf{Input}: $\theta, B^{-T}, \bar{\epsilon}, U, \mathcal{J},\theta^*,L,H$
\label{alg:PrecondKRK}
  \begin{algorithmic}[1]
  \State Draw $\xi\sim B^{-T}\mathcal{N}(0,I),\epsilon\sim\bar{\epsilon}\times\mathcal{U}_{[0.8,1]}$ 
  \State $\theta',v\gets \theta,\xi$
  \For{$i=\{1,\ldots,L\}$}
   \State $v\gets v-\frac{\epsilon}{2}\left(\mathcal{J}^{-1}\nabla_\theta U(\theta')-(\theta'-\theta^*)\right)$ 
    \State $\theta'\gets (\theta'-\theta^*)$
    \State $\theta',v\gets\theta'\cos(\epsilon)+v\sin(\epsilon),v\cos(\epsilon)-\theta'\sin(\epsilon)$ 
    \State $\theta'\gets \theta'+\theta^*$
    \State $v\gets v-\frac{\epsilon}{2}\left(\mathcal{J}^{-1}\nabla_\theta U(\theta')-(\theta'-\theta^*)\right)$
  \EndFor
\State $a\gets\min\left\{1,\exp(H(\theta,\xi)-H(\theta',v))\right\}$
\State Draw $\gamma\sim\mathcal{B}(a)$
\State $\theta\gets \gamma\theta'+(1-\gamma)\theta$
\end{algorithmic}
\end{algorithm}
\bigskip
Of course it is also possible to use the KDK Verlet \cref{alg:Verlet} with preconditioning ($M= \mathcal{J}$) (and $v$ replacing $p$). The resulting algorithm may be seen in \cref{alg:PrecondVerlet}.

\begin{algorithm}[h]
\caption{PrecondVerlet}
\textbf{Input}: $\theta, B^{-T}, \bar{\epsilon}, U, \mathcal{J},\theta^*,L,H$
\label{alg:PrecondVerlet}
  \begin{algorithmic}[1]
  \State Draw $\xi\sim B^{-T}\mathcal{N}(0,I),\epsilon\sim\bar{\epsilon}\times\mathcal{U}_{[0.8,1]}$
  \State $\theta',v\gets \theta,\xi$
  \For{$i=\{1,\ldots,L\}$}
   \State $v\gets v-\frac{\epsilon}{2}\mathcal{J}^{-1}\nabla_\theta U(\theta')$
    \State $\theta'\gets \theta'+\epsilon v$
   \State $v\gets v-\frac{\epsilon}{2}\mathcal{J}^{-1}\nabla_\theta U(\theta')$
  \EndFor
\State $a\gets\min\left\{1,\exp(H(\theta,\xi)-H(\theta',v))\right\}$
\State Draw $\gamma\sim\mathcal{B}(a)$
\State $\theta\gets \gamma\theta'+(1-\gamma)\theta$
\end{algorithmic}
\end{algorithm}

Applying these algorithms to the model problem \eqref{eq:MDUnconHam}, an analysis parallel to that carried out in Section~\ref{sec:shortcomings} shows that the decorrelation condition \eqref{eq:L} becomes, \emph{independently of $\sigma_1$ and $\sigma_2$}
$$
L \gtrapprox C/ \epsilon
$$
and the stability limits in \eqref{eq:verletlimit} and \eqref{eq:KRKlimit}  are now replaced, \emph{also independently of the values of $\sigma_1$ and $\sigma_2$,} by
$$
\epsilon \approx 2, \qquad \epsilon \approx \pi,
$$
for Algorithm~\ref{alg:PrecondVerlet} and Algorithm~\ref{alg:PrecondKRK} respectively. The stability limit for the PrecondRKR algorithm coincides with that of the PrecondKRK method.
(See also Appendix B.)

The idea of preconditioning is extremely old; to our best knowledge it goes back  to \citep{Bennett1975}. The algorithm in \citep{Girolami2011} may be regarded as a $\theta$-dependent preconditioning. For preconditioning in infinite dimensional problems see \citep{Beskos2011}.

\section{Numerical results}\label{sec:numerical}
In this section we test the following algorithms:
\begin{itemize}
\item Unconditioned Verlet: \cref{alg:Verlet} with $M=I$.
\item Unconditioned KRK: \cref{alg:UncondKRK}.
\item Preconditioned Verlet: \cref{alg:PrecondVerlet}.
\item Preconditioned KRK: \cref{alg:PrecondKRK}.
\item Preconditioned RKR: similar to \cref{alg:PrecondKRK} using a rotate-kick-rotate pattern instead of kick-rotate-kick.
\end{itemize}
The first two algorithms were compared in \citep{Shahbaba2014} and in fact we shall use the exact same logistic regression test problems used in that reference. If ${x}$ are the prediction variables and $y\in\{0,1\}$,  the likelihood for the test problems is ($\widetilde{{x}}=\left[1,{x}^T\right]^T,{\theta}=\left[\alpha,{\beta}^T\right]^T$)
\begin{equation}\label{eq:BILogRegLikelihood}
\mathcal{L}(\theta;x,y)=\prod_{i=1}^n\left(1+\exp(-{\theta}^T
\widetilde{{x}}_i)\right)^{-y}\left(1+\exp({\theta}^T\widetilde{{x}}_i)\right)^{y-1}.
\end{equation}

For the preconditioned integrators, we set $U_0$ as in \eqref{eq:U0} with $\theta^*$ given by the maximum a posteriori (MAP) estimation and $\mathcal{J}$ the Hessian at $\theta^*$.

For the two unconditioned integrators, we run the values of $L$ and $\epsilon$ chosen in \citep{Shahbaba2014} (this choice is labelled as A in the tables). Since in many cases the autocorrelation for the unconditioned methods is extremely large with those parameter values (see \cref{fig:Autocorrs}), we also present results for these methods with a principled choice of $T$ and $\epsilon$ (labelled as B in the tables). We take $T=\epsilon L={\pi}/({2\omega_{\min}})$, where $\omega_{\min}$ is the minimum eigenvalue of $\sqrt{D}$ given in \cref{eq:IntroExpMap}. In the case where the perturbation $U_1$ is absent, this choice of $T$ would decorrelate the least constrained component of $\theta$. We then set $\epsilon$ as large as possible to ensure an acceptance rate above 65\% \citep{Beskos2013Optimal}--- the stepsizes in the choice B are slightly smaller than the values used in \citep{Shahbaba2014}, and the durations $T$ are, for every dataset, larger. We are able thus to attain greater decorrelation, although at greater cost.
For the preconditioned methods, we set $T=\pi/2$, since this gives samples with 0 correlation in the case $U_1=0$, and then set the timestep $\epsilon$ as large as possible whilst ensuring the acceptance rate is above  65\%.

In every experiment we start the chain from the (numerically calculated) MAP estimate $\theta^\star$ of $\theta$ and acquire $N_s=5\times10^4$ samples. The autocorrelation times reported are calculated using the \texttt{emcee} function \texttt{integrated\_time} with the default value $c=5$ \citep{emcee}. We also estimated autocorrelation times using alternative methods \citep{Geyer1992,Neal1993,Sokal1997,Thompson2010}; the results obtained do not differ significantly from those reported in the tables.

Finally, note that values of $\bar\epsilon$ quoted in the tables are the maximum timestep that the algorithms operate with, since the randomisation follows $\epsilon\sim \bar{\epsilon}\times\mathcal{U}_{[0.8,1]}$. All code is available from the github repository \url{https://github.com/lshaw8317/SplitHMCRevisited}.

\subsection{Simulated Data}\label{sec:SimData}
We generate simulated data according to the same procedure and parameter values described in \citep{Shahbaba2014}.
The first step is to generate
${x}\sim\mathcal{N}(0,{\sigma}^2)$ with ${\sigma}^2=\mathrm{diag}\left\{\sigma_j^2: j=1\ldots,d-1\right\}$, where
$$
\sigma^2_j=\begin{cases}
25 &  j\leq 5\\
1 &  5<j\leq 10\\
0.04 &  j>10
\end{cases}.
$$
Then, we generate the true parameters $\hat{{\theta}}=[\alpha,{\beta}^T]^T$ with $\alpha\sim\mathcal{N}(0,\gamma^2)$ and the vector ${\beta}\in\mathbb{R}^{d-1}$ with independent components following $\beta_j\sim\mathcal{N}(0,\gamma^2),j=1,\ldots,d-1$, with $\gamma^2=1$.
Augmenting the data $\widetilde{{x}}_i=[1,{x}_i^T]^T$, from a given sample ${x}_i$, $y_i$ is then generated as a Bernoulli random variable $y_i\sim\mathcal{B}((1+\exp(-\hat{{\theta}}^T\widetilde{{x}}_i))^{-1})$.
In concreteness, a simulated data set $\{{x}_i,y_i\}_{i=1}^n$ with $n=10^4$ samples is generated, ${x}_i\in\mathbb{R}^{d-1}$ with $d-1=100$. The sampled parameters ${\theta}\in\mathbb{R}^d$ are assumed to have a prior $\mathcal{N}(0,\Sigma)$ with $\Sigma=\mathrm{diag}\left\{25: j=1\ldots,d\right\}$.
\begin{figure}[t]
    \begin{subfigure}{.45\textwidth}
	    \includegraphics[width=\textwidth]{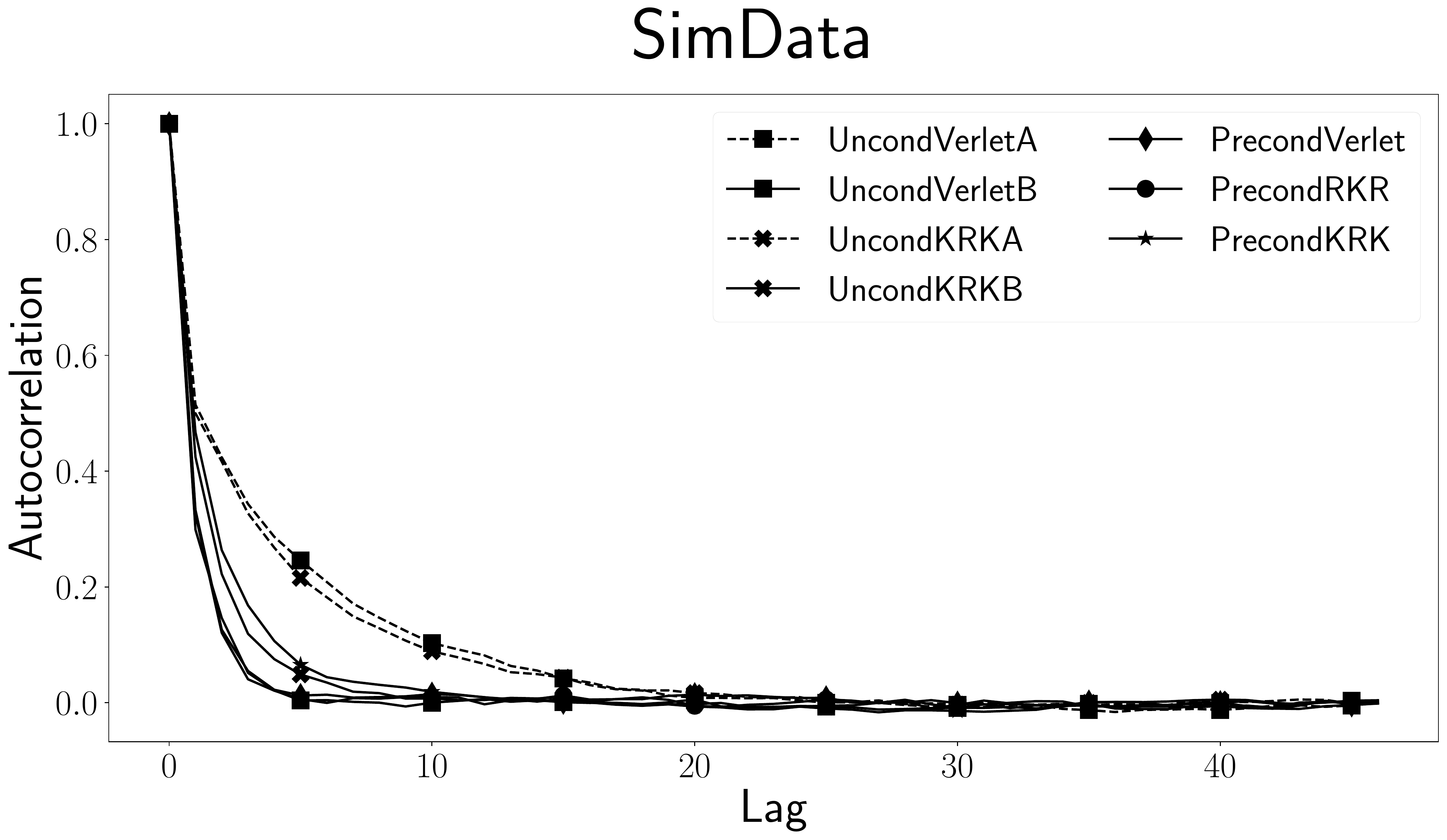}
	    \subcaption{}
    \end{subfigure}\quad\quad
    \begin{subfigure}{.45\textwidth}
	    \includegraphics[width=\textwidth]{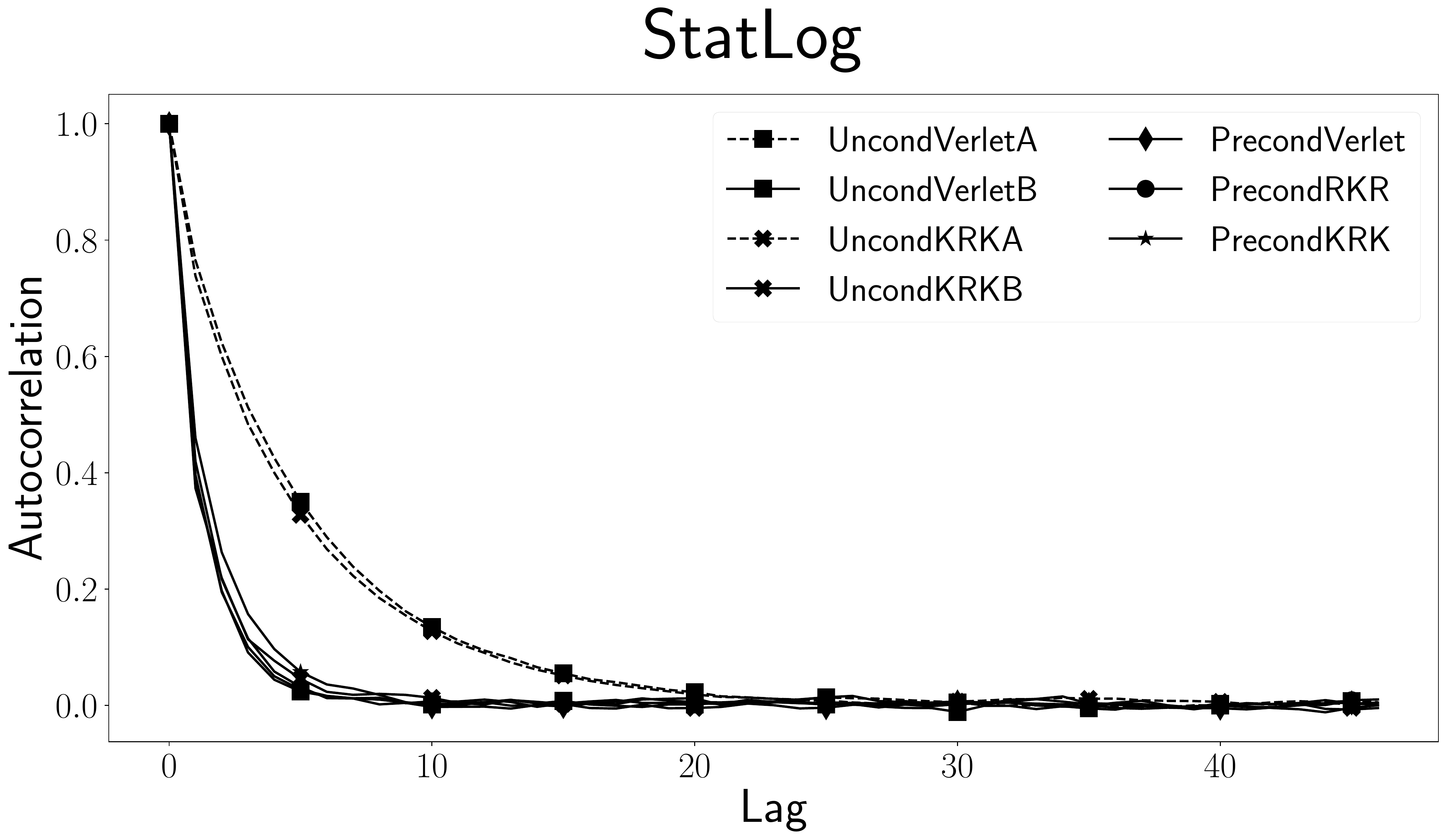}
	    \subcaption{}
    \end{subfigure}

    \begin{subfigure}{.45\textwidth}
	    \includegraphics[width=\textwidth]{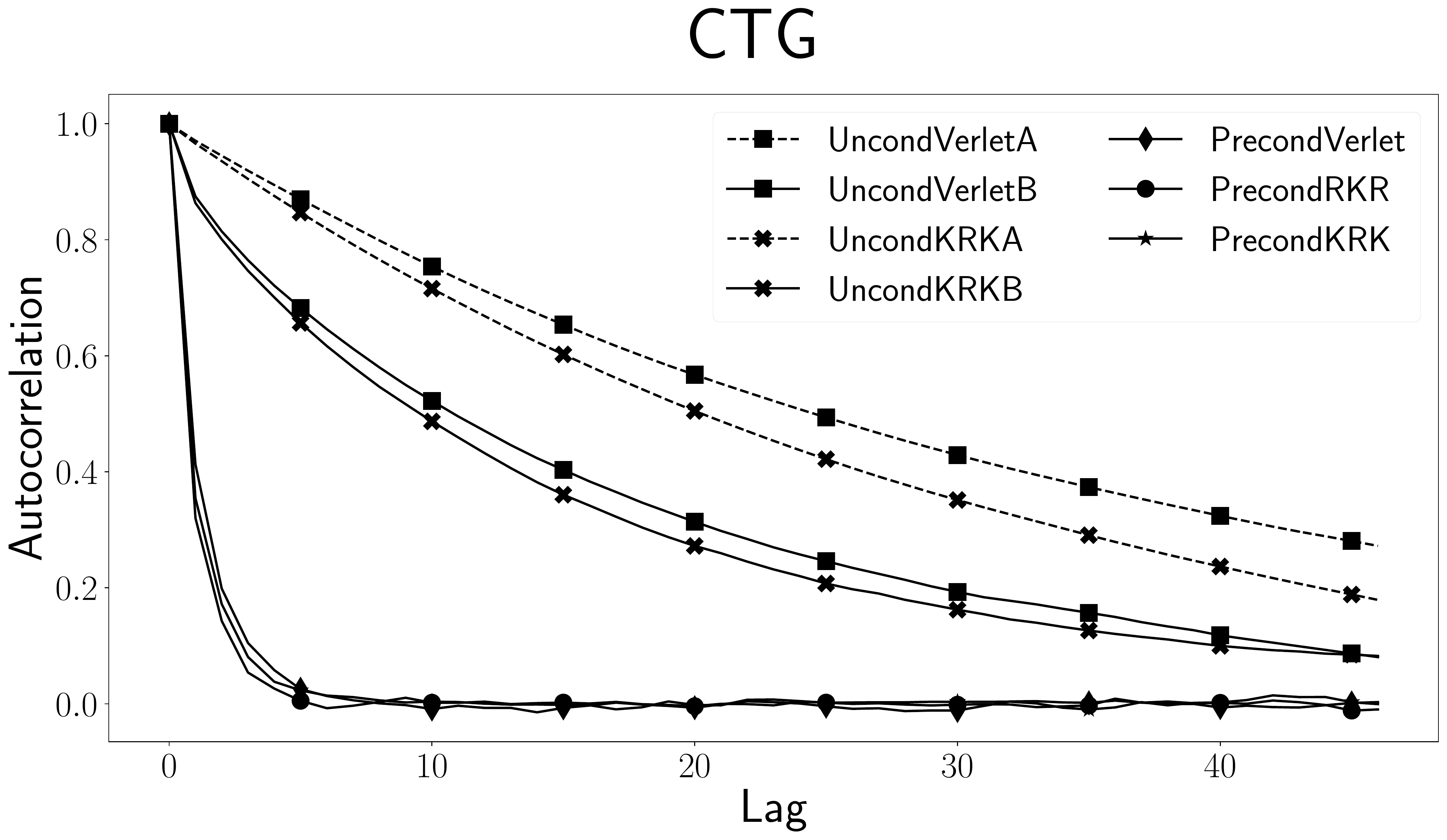}
	    \subcaption{}
    \end{subfigure}\quad\quad
    \begin{subfigure}{.45\textwidth}
	    \includegraphics[width=\textwidth]{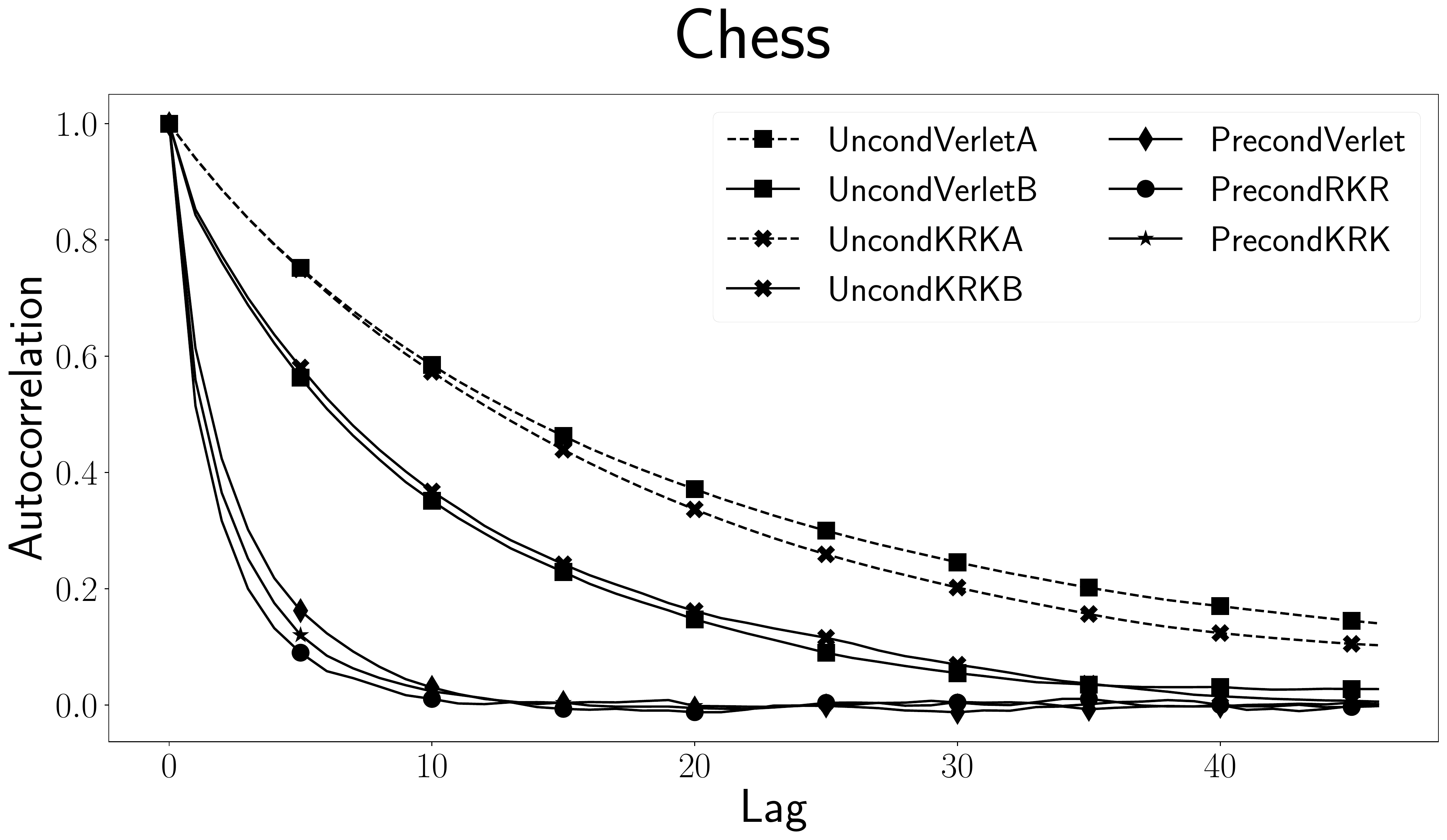}
	    \subcaption{}
    \end{subfigure}
    \caption{Autocorrelation function plots for the slowest moving component associated to the IAC $\tau_{\max}$ for each dataset. For the unconditioned methods, we show the principled choice B (solid line) and the choice  A from \citep{Shahbaba2014} (dotted). The values of $\epsilon$ and $T$ are as given in the tables.}
\label{fig:Autocorrs}
\end{figure}

Results are given in Table \ref{tab:simulated}. The second column gives the number $L$ of timesteps per proposal and the third the computational time $s$ (in milliseconds) required to
generate a single sample. The next columns give, for three observables, the products $\tau\times s$,
 with $\tau$ the integrated autocorrelation (IAC) time. These products measure the computational time to generate
 one independent sample. The notation $\tau_\ell$ refers to the observable $f(\theta)=\log(\mathcal{L}(\theta;x,y))$
 where $\mathcal L$ is the likelihood in \eqref{eq:BILogRegLikelihood}, and $\tau_{\theta^2}$ refers
  to $f(\theta)=\theta^T\theta$. The degree of correlation measured by $\tau_{\ell}$ is important in
  optimising the cost-accuracy ratio of predictions of $y$, while $\tau_{\theta^2}$ is relevant to estimating
   parameters of the distribution of $\theta$ \citep{Andrieu2003,GelmanBDA}.
 Following   \citep{Shahbaba2014}, we also examine the maximum IAC over all the Cartesian components
 of ${\theta}$, since we set the time $T$ in order to decorrelate the slowest-moving/least constrained
  component. Finally the last column provides the observed
   rate of acceptance.

 Comparing the values of $\tau\times s$ in the first four rows of the table shows the advantage, emphasized in \citep{Shahbaba2014}, of the $H_0+U_1$ \eqref{eq:main} over
  the kinetic/potential splitting: Unconditioned KRK operates with smaller values of $L$ than Unconditioned Verlet and the values of $\tau\times s$ are smaller for Unconditioned KRK than for unconditioned Verlet.
  However when comparing the results for Unconditioned Verlet A or B with those for Preconditioned Verlet,
it is apparent that the advantage of using the Hessian $\mathcal J$ to split $U=U_0+U_1$ with $M=I$ is much
smaller than the advantage of using $\mathcal J$ to precondition the integration while keeping the
kinetic/potential splitting.

The best performance is observed for the Preconditioned  KRK and RKR algorithms
that avail themselves of the Hessian \emph{both} to precondition and to use rotation instead of drift. Preconditioned RKR is clearly better
than its KRK counterpart (see Appendix B). For this problem, as shown in Appendix A, $U_1$ is in fact small and therefore the restrictions of the stepsize for the KRK integration are due to the stability reasons outlined in Section~\ref{sec:shortcomings}. In fact,
for the unconditioned algorithms, the stepsize $\bar \epsilon_{KRK}=0.03$ is not substantially larger than $\bar \epsilon_{Verlet}=0.015$, in agreement with the analysis presented in that section.

 The need to use large values of $L$ in the unconditioned integration stems, as discussed above, from the coexistence of large differences between the frequencies of the harmonic oscillators. In this problem the minimum and maximum frequencies are $\omega_{\min}=2.6,\omega_{\max}=105.0$.
 \begin{table}[h!]
\centering
\begin{tabular}{|c|c|c|c|c|c|c|}
\hline
& $L$ & $s$ [ms] & $\tau_{\ell}\times s$  & $\tau_{\theta^2}\times s$ &$\tau_{\max}\times s$ & AP \\ \hline\hline
UncondVerlet A& 20 & 4.70 & $3.5\times s=16.5$ & $11.4\times s=53.6$ & $7.0\times s=32.9$  & 0.69 \\ \hline
UncondVerlet B& 40 & 8.49 & $3.7\times s=31.4$ & $2.6\times s=22.1$ & $2.0\times s=17.0$ & 0.68\\ \hline
UncondKRK A& 10 & 3.04 & $3.4\times s=10.3$ & $11.1\times s=33.8$ & $6.6\times s=20.1$  & 0.76\\ \hline
UncondKRK B& 20 & 5.27 & $3.9\times s=20.5$ & $3.3\times s=17.4$ & $3.0\times s=15.8$ & 0.69\\ \hline
PrecondVerlet& 3 & 1.60 & $2.5\times s=4.0$ & $2.3\times s=3.7$ & $2.3\times s=3.7$ & 0.79\\ \hline
PrecondKRK& 1 & 1.22 & $2.8\times s=3.4$ & $3.4\times s=4.2$ & $3.5\times s=4.3$ & 0.75\\ \hline
PrecondRKR& 1 & 0.99 & $1.6\times s=1.6$ & $2.1\times s=2.1$ & $2.1\times s=2.1$ & 0.87\\ \hline
\end{tabular}
\caption{SimData: For methods labelled A, parameters from \citep{Shahbaba2014}: $T=0.3$, $\bar{\epsilon}_{Verlet}=0.015$, $\bar{\epsilon}_{UKRK}=0.03$. For the unconditioned methods labelled B, $T=\pi/2\omega_{\min}=0.6$, and $\bar{\epsilon}_{Verlet}=0.015$,  $\bar{\epsilon}_{UKRK}=0.03$. For the preconditioned methods, $T=\pi/2$, and $\bar{\epsilon}_{Verlet}=T/3\approx 0.52$; the other preconditioned methods operate with $\bar{\epsilon}_{Precon}=T\approx 1.57$. }
\label{tab:simulated}
\end{table}

\subsection{Real Data}
The three real datasets considered in \citep{Shahbaba2014}, StatLog, CTG and Chess, are also examined, see Tables 2--4. For the StatLog and CTG datasets with the unconditioned Hamiltonian, KRK does not really provide an improvement on Verlet. In all three datasets, the preconditioned integrators clearly outperform the unconditioned counterparts. Of the three preconditioned algorithms Verlet is the worst and RKR the best.

\paragraph{StatLog} Here, $n=4435$, $d-1=36$. The frequencies are $\omega_{\min}=0.5,\omega_{\max}=22.8$.
\begin{table}[h!]
\centering
\begin{tabular}{|c|c|c|c|c|c|c|}
\hline
& $L$ & $s$ [ms] & $\tau_{\ell}\times s$  & $\tau_{\theta^2}\times s$ &$\tau_{\max}\times s$ & AP \\ \hline\hline
UncondVerlet A& 20 & 1.99 & $5.5\times s=11.0$ & $5.8\times s=11.6$ & $9.8\times s=19.5$  & 0.69\\ \hline
UncondVerlet B& 40 & 3.34 & $7.6\times s=25.4$ & $2.5\times s=8.3$ & $2.6\times s=8.7$ & 0.64\\ \hline
UncondKRK A& 14 & 1.73 & $6.2\times s=10.7$ & $5.7\times s=9.9$ & $9.5\times s=16.5$  & 0.72\\ \hline
UncondKRK B& 28 & 2.79 & $8.7\times s=24.3$ & $2.9\times s=8.1$ & $2.9\times s=8.1$ & 0.65\\ \hline
PrecondVerlet& 3 & 0.64 & $2.5\times s=1.6$ & $2.6\times s=1.7$ & $2.7\times s=1.7$ & 0.88\\ \hline
PrecondKRK& 2 & 0.60 & $2.9\times s=1.7$ & $3.2\times s=1.9$ & $3.3\times s=2.0$ & 0.88\\ \hline
PrecondRKR& 2 & 0.53 & $2.3\times s=1.2$ & $2.5\times s=1.3$ & $2.7\times s=1.4$ & 0.94\\ \hline
\end{tabular}
\caption{StatLog: For methods labelled A, parameters are from \citep{Shahbaba2014}: $T=1.6$, $\bar{\epsilon}_{Verlet}=0.08$, $\bar{\epsilon}_{UKRK}=0.114$. For the unconditioned methods labelled B, $T=\pi/2\omega_{\min}=3.26$, and $\bar{\epsilon}_{Verlet}=0.08$, $\bar{\epsilon} _{UKRK}= 0.114$.  For the preconditioned methods, $T=\pi/2$, and $\bar{\epsilon}_{Verlet}=T/3$; the other preconditioned methods operate with $\bar{\epsilon}_{Precon}=T/2$.}
\end{table}

\paragraph{CTG}  Here, $n=2126$, $d-1=21$. The frequencies are $\omega_{\min}=0.2,\omega_{\max}=23.9$.

\begin{table}[h!]
\centering
\begin{tabular}{|c|c|c|c|c|c|c|}
\hline
& $L$ & $s$ [ms] & $\tau_{\ell}\times s$  & $\tau_{\theta^2}\times s$ &$\tau_{\max}\times s$ & AP \\ \hline\hline
UncondVerlet A& 20 & 1.06 & $5.9\times s=6.2$ & $20.1\times s=21.2$ & $80.3\times s=84.8$  & 0.69\\ \hline
UncondVerlet B& 98 & 4.28 & $6.1\times s=26.1$ & $5.1\times s=21.8$ & $36.0\times s=154.2$ & 0.64\\ \hline
UncondKRK A& 13 & 0.95 & $6.5\times s=6.2$ & $17.9\times s=17.0$ & $53.0\times s=50.3$  & 0.77\\ \hline
UncondKRK B& 66 & 3.89 & $6.1\times s=23.7$ & $5.1\times s=19.8$ & $37.2\times s=144.6$ & 0.65\\ \hline
PrecondVerlet& 2 & 0.36 & $2.6\times s=0.9$ & $2.1\times s=0.7$ & $2.6\times s=0.9$ & 0.76\\ \hline
PrecondKRK& 2 & 0.41 & $1.8\times s=0.7$ & $1.8\times s=0.7$ & $2.4\times s=1.0$ & 0.90\\ \hline
PrecondRKR& 2 & 0.35 & $1.9\times s=0.7$ & $1.7\times s=0.6$ & $2.1\times s=0.7$ & 0.93\\ \hline
\end{tabular}
\caption{CTG: For runs labelled A, parameters are from \citep{Shahbaba2014}: $T=1.6$, $\bar{\epsilon}_{Verlet}=0.08$, $\bar{\epsilon}_{UKRK}=0.123$. For the unconditioned runs labelled B, $T=\pi/2\omega_{\min}=7.85$, and $\bar{\epsilon}_{Verlet}=0.08$,  $\bar{\epsilon}_{UKRK}=0.118$. For the preconditioned methods, $T=\pi/2$, and $\bar{\epsilon}=T/2$. }
\end{table}

\paragraph{Chess} Here, $n=3196$, $d-1=36$. The frequencies are $\omega_{\min}=0.3,\omega_{\max}=22.3$.

\begin{table}[h!]
\centering
\begin{tabular}{|c|c|c|c|c|c|c|}
\hline
& $L$ & $s$ [ms] & $\tau_{\ell}\times s$  & $\tau_{\theta^2}\times s$ &$\tau_{\max}\times s$ & AP \\ \hline\hline
UncondVerlet A& 20 & 1.52 & $12.2\times s=18.5$ & $18.9\times s=28.6$ & $42.3\times s=64.1$  & 0.62\\ \hline
UncondVerlet B& 65 & 4.20 & $3.6\times s=15.1$ & $1.5\times s=6.3$ & $19.9\times s=83.6$ & 0.68\\ \hline
UncondKRK A& 9 & 0.90 & $13.3\times s=11.9$ & $21.3\times s=19.1$ & $37.7\times s=33.8$  & 0.72\\ \hline
UncondKRK B& 40 & 3.31 & $4.1\times s=13.6$ & $1.9\times s=6.3$ & $22.1\times s=73.1$ & 0.64\\ \hline
PrecondVerlet& 2 & 0.46 & $2.6\times s=1.2$ & $3.1\times s=1.4$ & $5.2\times s=2.4$ & 0.63\\ \hline
PrecondKRK& 2 & 0.50 & $1.6\times s=0.8$ & $2.5\times s=1.2$ & $4.6\times s=2.3$ & 0.81\\ \hline
PrecondRKR& 2 & 0.44 & $1.6\times s=0.7$ & $2.2\times s=1.0$ & $3.8\times s=1.7$ & 0.85\\ \hline
\end{tabular}
\caption{Chess: For runs labelled A, parameters are from \citep{Shahbaba2014}: $T=1.8$, $\bar{\epsilon}_{Verlet}=0.09$, $\bar{\epsilon}_{UKRK}=0.2$.
For runs labelled B, $T=\pi/2\omega_{\min}=5.71$, and $\bar{\epsilon}_{Verlet}=0.087$,  $\bar{\epsilon}_{UKRK}=0.142$. For the preconditioned methods, $T=\pi/2$, and $\bar{\epsilon}=T/2$. }
\end{table}

\begin{appendices}

\section{Bernstein-von Mises theorem}

From the Bernstein-von Mises theorem (see e.g.\ section 10.2 in \citep{VanderVaart2000}), as the size of the
dataset $n$ increases unboundedly, the posterior distribution $\pi(\theta\mid x,y)$ becomes dominated by the
likelihood and is asymptotically  Gaussian; more precisely  $
\mathcal{N}(\hat{\theta},n^{-1}\mathcal{I}_F(\hat{\theta})^{-1}) $, where $\hat \theta$ represents the true
value and $\mathcal{I}_F$ denotes the Fisher information matrix. This observation shows that, at least for $n$
large, approximating the potential $\propto\exp(-U(\theta))$ by a Gaussian $\propto\exp(-U_0(\theta))$ with
mean $\theta^*$  as in Section \ref{sec:SimData} is meaningful.

\begin{figure}[h]
\includegraphics[width=\textwidth]{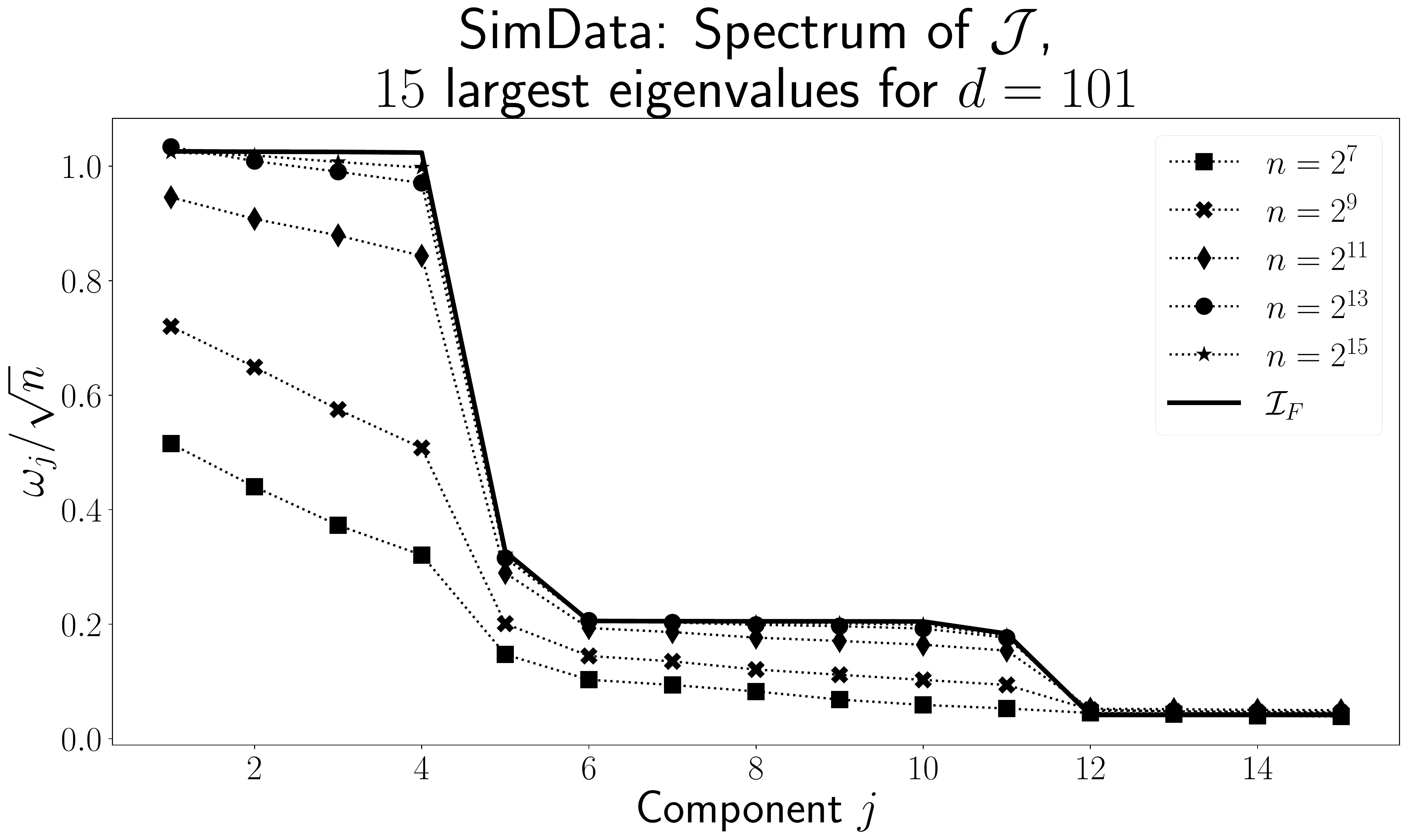}
\caption{The ordered rotation frequencies ($\omega_j=\sqrt{\lambda_j}$, with $\lambda_j$ an
eigenvalue of $\mathcal{J}$) scaled by $\sqrt{n}$ as the number of data points $n$ increases
 converge to the square roots of the eigenvalues of the (estimated) Fisher information matrix.}
\label{fig:SimDataSpectrum}
\end{figure}

\begin{figure}[h]
\includegraphics[width=\textwidth]{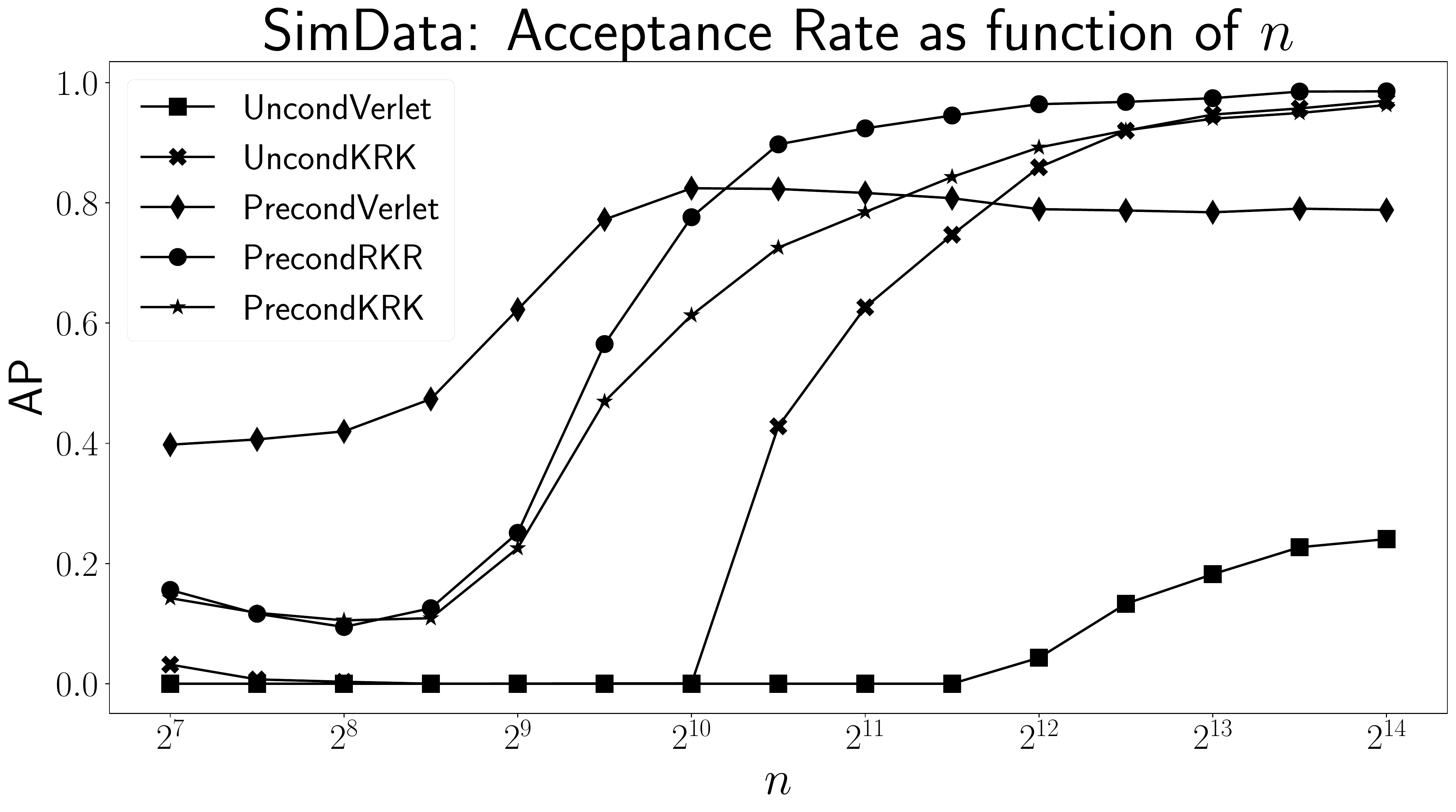}
\caption{As $n$ varies the algorithms are run with a fixed number $L$ of timesteps per proposal. For methods using
rotations the acceptance rate approaches $100\%$ as $n\uparrow\infty$.}
\label{fig:SimDatanrange}
\end{figure}

An illustration of the Bernstein-von Mises theorem is provided in Figure~\ref{fig:SimDataSpectrum} that
corresponds to the simulated data problem described in Section \ref{sec:SimData}. As the number of data points
increases from $2^7=128$ to $2^{14}= 32,768$, the scaled values $\omega_j/\sqrt{n}$ where $\omega_j^2$ are the
eigenvalues of the numerically calculated Hessian $\mathcal{J}(\theta^*)$ that we use in $U_0$ converge to the
square roots of the eigenvalues of the Monte Carlo estimation of the Fisher information matrix
$\mathcal{I}_F(\hat\theta)$ calculated using the true parameter values and the randomly generated ${ x}_i$.

A further illustration is provided in Figure~\ref{fig:SimDatanrange} where again the number of
data points  increases from $2^7=128$ to $2^{14}= 32,768$. The following parameter values are used:
\begin{itemize}
\item The preconditioned algorithms, where solutions of the Hamiltonian $H_0$ are periodic
with period $2\pi$, have $T=\pi/2$. For the KRK and RKR splittings we take two timesteps per proposal, i.e.\ $L=2$ and for the preconditioned Verlet, $L=3$.
\item For the unconditioned algorithms we set $T=(\pi/2)/\omega_{\min}$. i.e.\ a quarter of the largest period present in the solutions of $H_0$. Both Verlet and UncondKRK are operated with $L=30$ timesteps per proposal.
\end{itemize}
Note that since, as $n$ varies the value of $L$ for each algorithm remains constant,  the number of
evaluations of $\nabla U_1$ (for methods with rotations) or $\nabla U$ (for the Verlet integrator) remains
constant. The figure shows that, as $n$ increases, the acceptance rate for the methods Unconditioned KRK,
Preconditioned KRK, and Preconditioned RKR based on the splitting \eqref{eq:main} approaches $100\%$. These methods are exact when $U_1 = 0$ and
$\exp(-U)$ coincides with the Gaussian $\exp(-U_0)$ and therefore have smaller energy errors/larger acceptance
rates as $n$ increases. On the other hand the integrators based on the kinetic/potential splitting are not
exact when the potential $U$ is quadratic, and, correspondingly, we see that the acceptance rate does not
approach $100\%$ as $n\uparrow \infty$.

\section{Integrating the  preconditioned Hamilton equations}

 KRK and RKR are two possible reversible, symplectic integrators for the equations of motion corresponding to the preconditioned Hamiltonian \eqref{eq:IntroPrecondHam0}, but many others are of course possible. In this Appendix we  present a methodology to choose between different integrators. The material parallels an approach  suggested in \citep{Blanes2014} to choose between integrators for the kinetic/potential splitting; an approach that has been followed by a number of authors (see \citep{Blanes2021} for an extensive list of references). The methodology is based on using a Gaussian model distribution to discriminate between alternative algorithms, but, as shown in \citep{Calvo2021}, is very successful in predicting which algorithms will perform well for general distributions.

 To study the preconditioned $H_0+U_1$ splitting,  we select the model one-dimensional problem
\begin{equation}
H(\theta,p)=\frac{1}{2}(p^2+\theta^2)+\frac{1}{2}\kappa \theta^2.\label{eq:ModelProblem}
\end{equation}
We assume that $\kappa>-1$ so that the potential energy $(1/2) \theta^2+ (\kappa/2)\theta^2$ is positive definite. The application of one step (of length $\epsilon$) of an integrator for this problem in all practical contexts takes the linear form
\begin{equation}\label{eq:ModelProblemGenMatrix1}
     \begin{pmatrix}
    \theta_{n+1}\\ p_{n+1}
    \end{pmatrix}=
    {M}_{\kappa,\epsilon}
    \begin{pmatrix}
    \theta_{n}\\ p_{n}
    \end{pmatrix},\quad\quad {M}_{\kappa,\epsilon}=\begin{bmatrix}
    A_{\kappa,\epsilon} & B_{\kappa,\epsilon}\\
    C_{\kappa,\epsilon}& D_{\kappa,\epsilon}
    \end{bmatrix}.
\end{equation}
From \cref{eq:ModelProblemGenMatrix1}, it is clear that an integration leg of length $T=\epsilon L$ (with initial condition $\theta(0)=\theta_0,p(0)=p_0$) is given by
$$
     \begin{pmatrix}
    \theta_{L}\\ p_{L}
    \end{pmatrix}=
    {M}_{\kappa,\epsilon}^L
    \begin{pmatrix}
    \theta_{0}\\ p_{0}
    \end{pmatrix}.
$$

We now apply two restrictions to the integration matrix in \cref{eq:ModelProblemGenMatrix1}. Reversibility imposes that $A_{\kappa,\epsilon}=D_{\kappa,\epsilon}$; symplecticity (in one dimension equivalent to volume-preservation) implies that \citep {Blanes2014}
\begin{equation}\label{eq:ModelProblemSymplecticCondition}
    \det({M}_{\kappa,\epsilon})=A_{\kappa,\epsilon}^2-B_{\kappa,\epsilon}C_{\kappa,\epsilon}=1.
\end{equation}
The eigenvalues of the matrix ${M}_{\kappa,\epsilon}$ are then
$$
    \lambda=A_{\kappa,\epsilon}\pm\sqrt{A_{\kappa,\epsilon}^2-1},
$$
which shows that there are three cases:
\begin{enumerate}
    \item $\lvert A_{\kappa,\epsilon}\rvert>1$. For one of the eigenvalues, $\lvert\lambda\rvert>1$ and so the integration is unstable.
    \item $\lvert A_{\kappa,\epsilon}\rvert<1$. The integration is stable as both eigenvalues have magnitude $1$.
    \item $\lvert A_{\kappa,\epsilon}\rvert=1$. The symplectic condition \cref{eq:ModelProblemSymplecticCondition} necessarily implies $B_{\kappa,\epsilon}C_{\kappa,\epsilon}=0$, which gives two sub-cases:
        \begin{enumerate}
            \item $\lvert B_{\kappa,\epsilon}\rvert+\lvert C_{\kappa,\epsilon}\rvert=0$. The matrix ${M}_{\kappa,\epsilon}=\pm{I}$ and the integration is stable\label{item:ModelProblemStable2}.
            \item $\lvert B_{\kappa,\epsilon}\rvert+\lvert C_{\kappa,\epsilon}\rvert\neq0$. Then, if $B_{\kappa,\epsilon}\neq0$,
            $${M}^L_{\kappa,\epsilon}=\begin{bmatrix}
            A^L& LA^{L-1}B\\0&A^L
            \end{bmatrix}$$ and the integration is (weakly) unstable. Similarly there is weak instability if instead $C_{\kappa,\epsilon}\neq0$
        \end{enumerate}
\end{enumerate}

Thus for stable integration, one may find $\eta_{\kappa,\epsilon}$ such that $A_{\kappa,\epsilon}=\cos(\eta_{\kappa,\epsilon})\in[-1,1]$; in addition we define $\chi_{\kappa,\epsilon}=B_{\kappa,\epsilon}/\sin(\eta_{\kappa,\epsilon})$ for $\sin(\eta_{\kappa,\epsilon})\neq0$ and let
$\chi_{\kappa,\epsilon}$ be arbitrary if $\sin(\eta_{\kappa,\epsilon})=0$. In this way, for the model problem, all stable, symplectic integrations have a propagation matrix of the form
\begin{equation}\label{eq:ModelProblemGenMatrix3}
{M}_{\kappa,\epsilon}=\begin{bmatrix}
    \cos(\eta_{\kappa,\epsilon}) & \chi_{\kappa,\epsilon}\sin(\eta_{\kappa,\epsilon})\\
    -\chi_{\kappa,\epsilon}^{-1}\sin(\eta_{\kappa,\epsilon})& \cos(\eta_{\kappa,\epsilon})
    \end{bmatrix}.
\end{equation}

We now state a lemma analogue of Proposition 4.3 in \citep{Blanes2014}.
\begin{lemma}\label{thm:ModelProblemEnergyError}
Denote $A=\cos(L\eta_{\kappa,\epsilon}), B=\chi_{\kappa,\epsilon}\sin(L\eta_{\kappa,\epsilon}), C=-\chi^{-1}_{\kappa,\epsilon}\sin(L\eta_{\kappa,\epsilon})$. Given the initial conditions $\theta_0,p_0$, and integrating the dynamics of the Hamiltonian of the model problem \cref{eq:ModelProblem} using the integrator in \cref{eq:ModelProblemGenMatrix3} for $L$ steps to give new values of $\theta_L,p_L$, the energy error may be expressed as:
\begin{equation}\label{eq:ModelProblemEnergyError}
    \Delta\equiv H(\theta_L,p_L)-H(\theta_0,p_0)=\frac{1}{2}\left(C+(1+\kappa)B\right)\left(C\theta_0^2+2A\theta_0p_0+Bp_0^2\right).
\end{equation}
\end{lemma}
\begin{proof}
Applying the symplectic condition \cref{eq:ModelProblemSymplecticCondition}, the energy error $\Delta\equiv H(\theta_L,p_L)-H(\theta_0,p_0)$ then follows
\begin{align*}
    2\Delta &= p_L^2+(1+\kappa)\theta_L^2-p_0^2-(1+\kappa)\theta_0^2\\
    &=\left(C\theta_0+Ap_0\right)^2+(1+\kappa)\left(Bp_0+A\theta_0\right)^2-p_0^2-(1+\kappa)\theta_0^2\\
    &=(C^2+(A^2-1)(1+\kappa))\theta_0^2+2A(C+(1+\kappa)B)\theta_0p_0+(A^2-1+(1+\kappa)B^2)p_0^2\\
    &=\left(C+(1+\kappa)B\right)\left(C\theta_0^2+2A\theta_0p_0+Bp_0^2\right).
\end{align*}
\end{proof}
\begin{theorem}\label{cor:ModelProblemEnergyError}
With the notation of the lemma, assume that the initial conditions $\theta_0\sim\mathcal{N}(0,{1}/({1+\kappa}))$, $p_0\sim\mathcal{N}(0,1)$ are (independently) distributed according to their stationary distributions corresponding to the Hamiltonian for the model problem \cref{eq:ModelProblem}. Then the expected energy error follows
$$
    \mathbb{E}[\Delta]=\sin^2(L\eta_{\kappa,\epsilon})\rho(\epsilon,\kappa)\leq\rho(\epsilon,\kappa),
$$
where $\rho$ is given by
\begin{equation}\label{eq:ModelProblemRho}
    \rho(\epsilon,\kappa)=\frac{1}{2}\left(\sqrt{1+\kappa}\chi_{\kappa,\epsilon}-\frac{1}{\sqrt{1+\kappa}\chi_{\kappa,\epsilon}}\right)^2=\frac{\left(C_{\kappa,\epsilon}+(1+\kappa)B_{\kappa,\epsilon}\right)^2}{2(1+\kappa)(1-A_{\kappa,\epsilon}^2)}.
\end{equation}

\begin{proof}
Since $\mathbb{E}[\theta_0p_0]=\mathbb{E}[\theta_0]\mathbb{E}[p_0]=0$ and $\mathbb{E}[p_0^2]=1, \mathbb{E}[\theta_0^2]={1}/(1+\kappa)$, the expectation of \cref{eq:ModelProblemEnergyError} is
$$
    \mathbb{E}[\Delta]=\frac{1}{2}\left(C+(1+\kappa)B\right)\left(\frac{C}{1+\kappa}+B\right)=\frac{1}{2}\left(\frac{C}{\sqrt{1+\kappa}}+\sqrt{1+\kappa}B\right)^2.
$$
Substituting the expressions for $B,C$ from the definitions in the theorem above into the last display and dropping the subscripts to give $\chi=\chi_{\kappa,\epsilon}$, $s=\sin(L\eta_{\kappa,\epsilon})$ and $c=\cos(L\eta_{\kappa,\epsilon})$ gives
\begin{equation*}
    \mathbb{E}[\Delta]=\frac{1}{2}s^2\left(\frac{1}{(1+\kappa)\chi^2}+(1+\kappa)\chi^2-2\right)=\sin^2(L\eta_{\kappa,\epsilon})\rho(\epsilon,\kappa).
\end{equation*}
\end{proof}
\end{theorem}

Since $\eta_{\kappa,\epsilon}$ and $\chi_{\kappa,\epsilon}$ depend  on the integrator \cref{eq:ModelProblemGenMatrix3}, the  $\rho$ function also depends on the integrator. Note that $\rho$ does not change with $L$. For the model problem, integrators with smaller $\rho$ lead to smaller averaged energy errors at stationarity of the chain and therefore to smaller empirical rejection rates. By diagonalization it is easily shown as in \citep{Blanes2014} that the same is true for \emph{all Gaussian targets} $\propto \exp(-U(\theta))$, $U = U_0+U_1$. This suggests that,
all other things being equal, integrators with smaller $\rho$ should be preferred (see a full discussion in\ \citep{Calvo2021}).

\subsection{KRK vs. RKR}

For the KRK integration, we find (similarly to Section~\ref{sec:shortcomings}) that a
stable integration  requires $-1<\cos(\epsilon)-\epsilon\kappa\sin(\epsilon)/2<1$, so that, the stability limit for $\kappa>0$ is
\begin{equation}\label{eq:SpecificSchemesRKRStab}
    \epsilon<\frac{2\cot(\epsilon/2)}{\kappa}.
\end{equation}
Note that $\epsilon<\pi$ for any  value of $\kappa>0$. For $-1<\kappa<0$, the stability limit is $\epsilon<\pi$.
Application of the formula \cref{eq:ModelProblemRho} gives the $\rho$ function of the integrator  as
$$
    \rho^{[KRK]}(\epsilon,\kappa)=\frac{\kappa^2\csc(\epsilon)(-4 \epsilon\cos(\epsilon)+(4+\kappa \epsilon^2)\sin(\epsilon))^2}
    {8(1+\kappa)(4\kappa \epsilon\cos(\epsilon)+(4-\kappa^2\epsilon^2)\sin(\epsilon))}.
$$
For $\kappa=0$, $\rho$ vanishes as  expected because then the integration is exact.

Similarly, for RKR,
stable integration  requires $-1<\cos(\epsilon)-\epsilon\kappa\sin(\epsilon)/2<1$, so that, the stability limit for $\kappa>0$ of RKR  is the same we found in \eqref{eq:SpecificSchemesRKRStab}.
Again for $-1<\kappa<0$, the stability limit is $\epsilon<\pi$, as for KRK.

Application of the formula \cref{eq:ModelProblemRho} gives the $\rho$ function of the RKR integrator as
$$
    \rho^{[RKR]}(\epsilon,\kappa)=\frac{\kappa^2\csc(\epsilon)(\kappa \epsilon\cos(\epsilon)+2\sin(\epsilon)-(2+\kappa)\epsilon)^2}
    {2(1+\kappa)(4\kappa \epsilon\cos(\epsilon)+(4-\kappa^2\epsilon^2)\sin(\epsilon))}.
$$

The following result implies  that for all Gaussian problems, at stationarity, RKR always leads to smaller energy errors/higher acceptance rates than KRK.
\begin{theorem}For each choice of $\kappa>-1$, $\kappa\neq 0$, and $\epsilon>0$ leading to a stable KRK or RKR integration
$$
\rho^{[RKR]}(\epsilon,\kappa) < \rho^{[KRK]}(\epsilon,\kappa).
$$
\end{theorem}
\begin{proof}
From the expressions for $\rho$ given above, we have to show that
$$
4(\kappa \epsilon\cos(\epsilon)+2\sin(\epsilon)-(2+\kappa)\epsilon)^2 <
(-4 \epsilon\cos(\epsilon)+(4+\kappa \epsilon^2)\sin(\epsilon))^2.
$$
It is therefore sufficient to show that
\begin{equation}\label{eq:first}
2(\kappa \epsilon\cos(\epsilon)+2\sin(\epsilon)-(2+\kappa)\epsilon) <
-4 \epsilon\cos(\epsilon)+(4+\kappa \epsilon^2)\sin(\epsilon)
\end{equation}
and
\begin{equation}\label{eq:second}
4 \epsilon\cos(\epsilon)-(4+\kappa \epsilon^2)\sin(\epsilon) < 2(\kappa \epsilon\cos(\epsilon)+2\sin(\epsilon)-(2+\kappa)\epsilon).
\end{equation}
The inequality \eqref{eq:first} may be rearranged as
$$
2\kappa \epsilon \big(\cos(\epsilon)-1\big) < 4\epsilon \big(1-\cos(\epsilon)\big)+\kappa\epsilon^2\sin(\epsilon),
$$
or
$$
-4\kappa \epsilon \sin^2(\epsilon/2) < 8 \epsilon \sin^2(\epsilon/2)+2\kappa \epsilon^2 \sin(\epsilon/2)\cos(\epsilon/2).
$$
Since for stable runs $\epsilon<\pi$, so that $\sin^2(\epsilon/2)>0$, the last display is equivalent to
$$
-4\kappa < 8 +2\kappa \epsilon \cot(\epsilon/2)
$$
or
$$
-4 <\kappa \big(\epsilon \cot(\epsilon/2)+2\big).
$$
For $0<\epsilon <\pi$, $\epsilon \cot(\epsilon/2)+2$ takes values between 4 and 2 and therefore the last inequality certainly holds for each $\kappa>-1$.

The inequality \eqref{eq:second} may be similarly rearranged as
$$
\kappa \epsilon \big(2-\epsilon \cot(\epsilon/2)\big) < 4 \big(2-\epsilon \cot(\epsilon/2)\big)\cot(\epsilon/2)\big).
$$
Since, for $0<\epsilon <\pi$, $2-\epsilon \cot(\epsilon/2)>0$, we conclude that
\eqref{eq:second} is equivalent to
$$\kappa\epsilon < 4 \cot(\epsilon/2),
$$
a relation that, according to \eqref{eq:SpecificSchemesRKRStab}, holds for stable integrations with $\kappa>0$ and is trivially satisfied for $\kappa<0$, $\epsilon\in(0,\pi)$.
\end{proof}

\subsection{Multistage splittings}
In addition to the Strang formula \eqref{eq:VerletStrang} one may consider more sophisticated schemes
\begin{equation}\label{eq:multistage}
\psi_\epsilon =
\varphi^{[H_1]}_{a_1\epsilon}\circ
\varphi^{[H_2]}_{b_1\epsilon}\circ
\varphi^{[H_1]}_{a_2\epsilon}\circ
\cdots
\varphi^{[H_1]}_{a_{m-1}\epsilon}\circ
\varphi^{[H_2]}_{b_{m-1}\epsilon}\circ
\varphi^{[H_1]}_{a_m\epsilon}
\end{equation}
where $\sum a_j = \sum b_j = 1$. These integrators are always symplectic and in addition are time reversible if they are palindromic i.e. $a_i = a_{m-i+1}$, $i = 1,\dots, m$, $b_i = b_{m-i}$, $i = 1,\dots, m-1$. In the case of the kinetic/potential splitting of $H$,  integrators of the form \eqref{eq:multistage}  when used for HMC sampling may provide very large improvements on leapfrog/Verlet (see \citep{Calvo2021,Blanes2021} and their references).
For the preconditioned $H_0+U_1$ splitting in this paper, we have investigated extensively the existence of formulas of the format \eqref{eq:multistage} that improve on the RKR integrator based on the Strang recipe \eqref{eq:VerletStrang}. We proceeded in a way parallel to that followed in \citep{Blanes2014}. For fixed $m$,   $m=3$ or $m=4$, and a suitable range of values of $\kappa$ and $\epsilon$, we choose the values of $a_i$ and $b_i$ so as to minimize the function $\rho$ in Theorem \ref{cor:ModelProblemEnergyError}, thus minimizing the expected energy error at stationarity in the integration of the model problem. The outcome of our investigation was that, while we succeeded in finding formulas that improve on the Preconditioned RKR integrator, the improvements were minor and did not warrant the replacement of RKR by more sophisticated formulas.

\end{appendices}
\bibliographystyle{unsrt}

\end{document}